\documentclass[smallextended]{svjour3}
\usepackage{amsmath}
\usepackage{amssymb}
\usepackage{hyperref}
\usepackage{complexity}
\usepackage{empheq}
\usepackage{longtable}
\usepackage[table]{xcolor}

\DeclareMathOperator{\res}{\mathsf{res}}
\DeclareMathOperator{\rsc}{\mathsf{rsc}}
\DeclareMathOperator{\nonce}{\mathsf{nonce}}

\DeclareMathOperator{\GF}{\mathsf{GF}}
\DeclareMathOperator{\RTK}{\mathsf{RTK}}
\DeclareMathOperator{\RAS}{\mathsf{RAS}}
\DeclareMathOperator{\MAPK}{\mathsf{MAPK}}
\DeclareMathOperator{\PI3K}{\mathsf{PI3K}}
\DeclareMathOperator{\PIP3}{\mathsf{PIP3}}
\DeclareMathOperator{\FOXO3}{\mathsf{FOXO3}}
\DeclareMathOperator{\AKT}{\mathsf{AKT}}
\DeclareMathOperator{\cycE}{\mathsf{cycE/CDK2}}
\DeclareMathOperator{\Rb}{\mathsf{Rb}}
\DeclareMathOperator{\E2F}{\mathsf{E2F}}
\DeclareMathOperator{\TSC}{\mathsf{TSC}}
\DeclareMathOperator{\PRAS40}{\mathsf{PRAS40}}
\DeclareMathOperator{\mTORC1}{\mathsf{mTORC1}}
\DeclareMathOperator{\EIF4F}{\mathsf{EIF4F}}
\DeclareMathOperator{\S6K}{\mathsf{S6K}}
\DeclareMathOperator{\Pro}{\mathsf{Prolif}}
\DeclareMathOperator{\uPro}{\mathsf{UProlif}}

\title{Controllability of reaction systems}
\author{Sergiu Ivanov \and Ion Petre}
\institute{Sergiu Ivanov \at
IBISC, Universit\'{e} \'{E}vry, Universit\'{e} Paris-Saclay, France\\
\email{sergiu.ivanov@univ-evry.fr}
\and Ion Petre \at
Department of Mathematics and Statistics, University of Turku, Finland and
\at
National Institute for Research and Development in Biological Sciences, Romania\\
\email{ion.petre@utu.fi}
} 
\date{Received: date / Accepted: date}

\begin{document}

\maketitle
\begin{abstract}
	Controlling a dynamical system is the ability of changing its configuration arbitrarily through a suitable choice of inputs. It is a very well studied concept in control theory, with wide ranging applications in medicine, biology, social sciences, engineering. We introduce in this article the concept of controllability of reaction systems as the ability of transitioning between any two states through a suitable choice of context sequences. We show that the problem is \PSPACE-hard. We also introduce a model of oncogenic signalling based on reaction systems and use it to illustrate the intricacies of the controllability of reaction systems.

\end{abstract}


\section{Introduction}
Reaction systems are a biologically inspired model of computing
originally introduced in~\cite{Ehrenfeucht:2007aa}.  They capture two
fundamental interactions typically present between biochemical entities---activation and inhibition.
Reaction systems are dynamical systems: reactions transform a set of
reactants into a set of products provided that none of its inhibitors
are present, which are then transformed further into other products, etc.

The reactions in reaction systems are governed by two fundamental
principles: the \emph{threshold principle} and the
\emph{non-permanency principle}.  The threshold principle stipulates
that when a resource is available it is available in unlimited
amounts.  This defines reaction systems as a qualitative modelling
framework, whose states are sets of species, without any quantitative
information.  This also means that concurrency has to be modelled
explicitly, rather than implicitly via some intrinsic mechanisms of
the modelling device.  The non-permanency principle states that if a
resource is not explicitly sustained/produced by reactions, it will disappear.  The next
state of a reaction system only consists of the species explicitly
produced by the reactions enabled in the previous state.

Reaction systems are open systems: there is a notion of context that adds to the current state in each step of its dynamic process. The next state is produced by the reactions applied to the previous state plus the species added by the context. 

Since their introduction in 2007, two major research directions on
reaction systems have been established.  The first direction focuses
on their formal properties as a dynamical systems: sequences of
states~\cite{DBLP:journals/tcs/Salomaa12a}, \cite{cSalomaa_Arto12a}, power of
small systems for various size measures~\cite{Formenti2015}, \cite{jSalomaa_Arto13b}, \cite{cSalomaa_Arto14a}, \cite{jSalomaa_Arto15c}, \cite{TehAtanasiu2020}, cycles and attractors~\cite{Azimi:2017aa}, \cite{Dennunzio201996}, \cite{Formenti2014b}, connections to propositional
logic~\cite{DBLP:journals/ijfcs/Salomaa13}, etc.  The second major
direction of research on reaction systems consists in exploring their
potential as a modelling framework, in particular for biological
applications~\cite{jAzIaPe14a}, \cite{IF-RS}, \cite{CorolliMMBM12}.  This direction
sparked interest in model checking for reaction systems, i.e. formally
defining relevant properties, evaluating their complexity, and
designing algorithms for checking them~\cite{jAzGrIvMaPePo16a}, \cite{jAzGrIvPe15a}, \cite{Meski:2017aa}, \cite{Meski2015}.  These works revealed a
whole wealth of properties whose complexity ranges from polynomial to
\PSPACE-hard.

In this paper, we continue the study of potential applications of
reaction systems to biological and medical research, and we focus on
\emph{controllability}.  Intuitively, controllability of a dynamical
system is the ability of driving this system to any one of its states,
starting from any other state.  Controllability is a strong property,
which has attracted a lot of attention, especially in the case of
biological networks (e.g.,~\cite{KolchHGK2015}, \cite{LiuSB2011}).  Extensive
research has been conducted into the practical feasibility of
different variants of controllability for biological networks, and
exact and approximate algorithms for finding ways to drive them,
e.g.~\cite{NetControlTCBB2018}, \cite{Gao:2014aa}, \cite{Kanhaiya:2017aa}, \cite{NetControlGeneticFrontiers}.  These results have considerable
potential for applications.  For example,~\cite{BadhwarBagler2015}
identifies ``driver neurons which can provide full control over the
network'' governing the actions of the C.~elegans worm. Also, \cite{art-lsb11}, \cite{Gao:2014aa}, \cite{art-kcgp17}, \cite{MM3networks2020} discuss applications in drug repurposing and personalised medicine.

The main goal of the present paper is defining controllability for
reaction systems, and establishing some computational 
complexity evaluations.  Since reaction systems are intrinsically open
systems due to the context sequences governing their evolution,
introducing controllability is quite natural.  In addition to
conventional, unrestricted controllability, we further define
a restricted variant, similar to the notion typically used in network
control theory, i.e., target controllability~\cite{Gao:2014aa}, \cite{NetControlTCBB2018}.  We illustrate our definitions on a novel, oncogenic signalling model that we constructed based the Boolean model of~\cite{G.-T.-Zanudo:2018aa}.  We show that imposing restrictions on the context
sequences or on the allowed observables pushes the complexity of the controllability of reaction systems
to \PSPACE-hard.

This paper is structured as follows.  Section~\ref{sec:preliminaries}
recalls the basic notions, in particular reaction systems,
reachability, Boolean networks, and general controllability.
Section~\ref{sec:example} introduces the running example we will use
to illustrate the notions of controllability, and which is translated
from the Boolean model in~\cite{G.-T.-Zanudo:2018aa}.
Section~\ref{sec:control} defines controllability for reaction
systems, and Section~\ref{sec:target_control} defines target
controllability.  Finally, Section~\ref{sec:complexity} evaluates the
complexity of the decision problems associated with our definitions.


\section{Preliminaries}\label{sec:preliminaries}

\subsection{Reaction systems}

We recall in this section some of the basic concepts of reaction systems. The presentation only aims to fix the notation and is kept brief. For details we refer to \cite{Ehrenfeucht:2007aa}, \cite{Ehrenfeucht:2016aa}.

\paragraph{Basic definitions.} We introduce here the basic concepts around reaction systems. 
\begin{definition}
	Let $S$ be a finite (so-called background) set. 
	\begin{itemize}
		\item A \emph{reaction} in $S$ is a triplet $r=(R,I,P)$ such that  $R,I,P\subseteq S$, $R\cap I=\emptyset$. We call $R$, $I$, $P$ the \emph{reactant}, the \emph{inhibitor}, and the \emph{product} set of reaction $r$, resp. We also say that $R\cup I$ is the \emph{resource set} of reaction $r$ and denote it $R\cup I=\rsc(r)$. For a set $A$ of reactions, its resource set is $\rsc(A)=\cup_{r\in A}\rsc(r)$.
		\item Let $T\subseteq S$. We say that reaction $r=(R,I,P)$ is \emph{enabled} in $T$ if $R\subseteq T$ and $I\cap T=\emptyset$. The \emph{result} of reaction $r$ on $T$ is $\res_r(T)=P$ if $r$ is enabled on $T$ and it is $\res_r(T)=\emptyset$ otherwise. 
		\item Let $R$ be a set of reactions in $S$ and $T\subseteq S$. The result of $R$ on $T$ is $\res_R(T)=\cup_{r\in R}\res_r(T)$.
		\item A \emph{reaction system} $\mathcal{A}=(S,A)$ (over $S$) consists of a (finite) set of reactions $A$ in $S$. 
		\item The \emph{states} of reaction system $\mathcal{A}=(S,A)$ are the subsets of $S$. We say that state $V\subseteq S$ is \emph{reachable} in $\mathcal{A}$ if it is in the domain of $\res_A$, i.e.,  $V=\res_A(U)$, for some $U\subseteq S$.
		\item Let $\mathcal{A}=(S,A)$ be a reaction system and let $n\geq 1$. Let $\gamma=C_0,C_1,\ldots C_n\subseteq S$ be a sequence of so-called context sets. The \emph{interactive process} $\pi_\gamma(\mathcal{A})$ in $\mathcal{A}$ defined by $\gamma$ consists of two state sequences of length $n$ $\pi_\gamma(\mathcal{A})=(\delta_\gamma(\mathcal{A}),\tau_\gamma(\mathcal{A}))$ defined in the following way:
		
			\begin{itemize}
				\item $\delta_\gamma(\mathcal{A})=D_0,D_1,\ldots,D_n$ is the \emph{result sequence} of  $\pi_\gamma(\mathcal{A})$ and $\tau_\gamma(\mathcal{A})=W_0,W_1,\ldots,W_n$ is its \emph{state sequence};
				\item $D_0=\emptyset$, and $D_i=\res_\mathcal{A}(C_{i-1}\cup D_{i-1})$ for all $i\in\{1,\ldots,n\}$;
				\item $W_i=C_i\cup D_i$ for all $i\in\{0,1,\ldots,n\}$. $W_0$ is called the \emph{initial state} of $\pi_\gamma(\mathcal{A})$. We say that $\pi_\gamma(\mathcal{A})$ \emph{starts} in $W_0$ and \emph{ends} in $W_n$ and denote it $W_n=\res_\mathcal{A}^\gamma(W_0)$ (the \emph{result along context sequence $\gamma$}).
				\item We say that state $V$ is \emph{reachable} from state $U$ if there is a context sequence $\gamma$ such that $V=\res_\mathcal{A}^\gamma(U)$.
				\item If $\gamma$ consists of empty (context) sets only, then $\pi_\gamma(\mathcal{A})$ is called \emph{context-independent}. 
			\end{itemize}
	\end{itemize}
\end{definition}

The dynamic processes defined by interactive processes can be seen as state transition systems. We define the state transition graph of a reaction system as the graph having as nodes the states of the reaction system and the edges defined by the result function of the reaction system as follows. 

\begin{definition}
	Let $\mathcal{A}=(S,A)$ be a reaction system and $I\subseteq S$. The $I$-context graph of reaction system $\mathcal{A}$ is the graph $\mathcal{G}_\mathcal{A}^I=(\mathcal{P}(S),E)$, where the set of edges $E$ is $E=\{(X,Y)\mid \exists  C\subseteq I: \res_\mathcal{A}(X\cup C)=Y\}$. 
\end{definition}

In the definition of the context graph we restrict the context sets to be subsets of a given ``input'' set $I$ -- this will be useful when we introduce the notions of controllability for reaction systems. For $I=S$ there is no constraint on the context sets and we obtain the usual transition system associated to a reaction system, see \cite{Ehrenfeucht:2016aa}.

Model checking for reaction systems has been considered in a number of different setups, based on, e.g., temporal logic (\cite{Meski2015}, \cite{Meski:2017aa}, \cite{Meski2019}), and computational complexity (\cite{jAzGrIvPe15a}, \cite{jAzGrIvMaPePo16a}, \cite{Azimi:2017aa}). We recall here the result on the reachability problem \cite{ReachabilityRS}.

\begin{theorem}[\cite{ReachabilityRS}]
	Deciding if a state $V$ of a reaction system $\mathcal{A}$ is reachable from a state $U$ is \PSPACE-complete. 
\end{theorem}

\paragraph{Boolean networks and reaction systems.}
Multiple connections between Boolean functions and reaction systems
have been established, see
\cite{BarbutiBGLM18,ATourOfRS,Ehrenfeucht:2007aa}.  In what concerns
Boolean networks, different ways of defining the model exist, as for
example threshold Boolean networks (e.g.,
\cite{BarbutiBGLM18,VuongCIGT17}), or Boolean networks employing
propositional formulae (e.g., \cite{BianeD19,G.-T.-Zanudo:2018aa}).
We only consider here the Boolean networks employing propositional
formulae, working in synchronous mode.

We recall a standard way of simulating a single Boolean
function with reaction systems, proposed
in~\cite{ATourOfRS,Ehrenfeucht:2007aa}. Let \(f : \mathbb B^X \to \mathbb B\) be a Boolean function, \(f
    \notin X\), and \(\varphi\) a propositional formula in minimal disjunctive normal
form implementing \(f\): \(\varphi = \bigvee_i C_i\), where \(C_i\) are the
conjunctions appearing in \(\varphi\).  We will use the notation \(pos(C_i)\)
to refer to the set of variables appearing without a negation in
\(C_i\) and \(neg(C_i)\) to refer to the set of variables appearing with
a negation in \(C_i\). Motivated by our biological running example, we  also use  a new variable $\iota_f$ to $f$ that will be used to block the production of $f$ by inhibiting all reactions producing $f$. From the mathematical point of view, these extra inhibitors are not needed for the equivalence between Boolean networks and reaction systems. 

The following set of reactions corresponds to $\varphi$:
\[A_\varphi = \{(pos(C_i), neg(C_i)\cup\{\iota_f\}, \{f\}) \mid C_i \in \varphi\}.\]
The reaction system \(\mathcal A_\varphi = (X \cup \{f,\iota_f\}, A_\varphi)\) simulates \(f\) in the following sense.  Consider any truth assignment \(s : X \to \mathbb B\) and construct the corresponding subset of variables \(W_s = \{x \in X \mid
    s(x) = 1\}\), using \(s\) as a indicator function.  Then \(\res_{\mathcal A}(W_s) = \{f\}\) if and only if \(f(s) = 1\).  Indeed, \(\mathcal A_\varphi\) only
produces \(f\) on \(W_s\) if and only if there exists a conjunction \(C_i
    \in \varphi\) for which \(pos(C_i) \subseteq W_s\) and \(neg(C_i) \cap W_s = \emptyset\).

Consider now the Boolean network \(BN = (f_1, \dots, f_n)\) over the set of
variables \(X = \{x_1, \dots, x_n\}\), in which \(f_i\) is used to update
the variable \(x_i\).  The following set of
reactions is obtained by iterating the idea above for each Boolean function of the network:
\[A_{BN} = \{(pos(C), neg(C)\cup \{\iota_{x_i}\}, \{x_i\}) \mid C \in \varphi_i, 1 \leq i \leq n\},\]
where \(\varphi_i\) denotes the Boolean formula in minimal disjunctive
normal form implementing the update function \(f_i\).  The reaction
system \(\mathcal A_{BN} = (X, A_{BN})\) simulates the \emph{synchronous} model of
dynamics of the Boolean network \(BN\).

\subsection{Controllability of linear dynamical systems}

We introduce here briefly the controllability of linear dynamical systems. For a more detailed presentation we refer to \cite{NetControlTCBB2018}, \cite{Kanhaiya:2017aa}, \cite{NetControlGeneticFrontiers}. We only discuss here a few basic concepts to guide our definitions of the controllability of reaction systems. Intuitively, a linear dynamical system consists of a set of nodes (variables) influencing each other's dynamics through linear, one-source/one-target interactions. They can also be influenced through external, arbitrary interventions. The goal is to be able to change the configuration of the system from any initial state to any final state through a suitable choice of external interventions (that depend on the initial and desired final state). A system having this property is called controllable. A linear dynamical system is always trivially controllable by adding external interventions on all its nodes, that conveniently change its state as desired. The typical question to ask is one of optimisation: given a linear dynamical system, find the minimal set of external interventions making it controllable. We introduce these concepts formally in the following. 

A \emph{linear dynamical system} is a vector $x$ of functions $x:\mathbb{R}\rightarrow\mathbb{R}^n$, $n\geq 1$, defined as the solution of the system of ordinary differential equations 
\[\frac{dx(t)}{dt}=Ax(t),\]
for some fixed initial value $x(0)\in\mathbb{R}^n$. The matrix $A$ defining the dynamical system is an $n\times n$ real-valued matrix. The $(i,j)$ entry of matrix $A$ describes the influence of the $j$-th node of the system over its $i$-th node.

A linear dynamical system can also be influenced through an \emph{external contribution}, thought of as a parametric $m$-dimensional vector $u$ of real functions, influencing the $n$ nodes of the dynamical system through a matrix $B\in\mathbb{R}^{n\times m}$. In this case, the linear dynamical system is defined as the solution of the following system of ordinary differential equations:
\begin{equation}\label{eq-lds}
\frac{dx(t)}{dt}=Ax(t)+Bu(t),
\end{equation}
and it is called the $(A,B)$ linear dynamical system. 

One can also define a subset of so-called \emph{target nodes} of the dynamical system on which the behaviour of the system is observed: $T=\{t_1,t_2,\ldots,t_l\}\subseteq \{1,2,\ldots,n\}$, $1\leq l\leq n$, $t_i<t_j$ for $1\leq i<j\leq n$. The set of target nodes can be defined through its 0/1-valued characteristic matrix $C_T\in\mathbb{R}^{l\times n}$ defined as follows: $C_T(i,j)=1$ if and only if $t_i=j$. Obviously, if $T=\{1,2,\ldots,n\}$, then $C_T$ is the identity matrix. A \emph{targeted linear dynamical system} is defined by a triplet $(A,B,T)$. 

We say that a dynamical system $(A,B)$ is \emph{controllable} if for any $x(0)\in\mathbb{R}^n$ and any $\alpha\in\mathbb{R}^n$, there is an input function vector $u_{x(0),\alpha}:\mathbb{R}\rightarrow\mathbb{R}^n$ such that the solution $x$ of \eqref{eq-lds} satisfies the property $x(\tau)=\alpha$, for some $\tau\geq 0$. 

We say that a targeted dynamical system $(A,B,T)$ is \emph{target controllable} if for any $x(0)\in\mathbb{R}^n$ and any $\gamma\in\mathbb{R}^l$, there is an input function vector $v_{x(0),\alpha}:\mathbb{R}\rightarrow\mathbb{R}^n$ such that the solution $x$ of \eqref{eq-lds} satisfies the property $C_Tx(\tau)=\gamma$, for some $\tau\geq 0$. In other words, the solution eventually matches $\gamma$ on its $T$-components. Obviously, for $T=\{1,2,\ldots,n\}$, target controllability is identical to controllability. 

The (target) controllability problem has an elegant algebraic characterisation known as Kalman's condition. 

\begin{theorem}[Kalman's condition \cite{Kalman1963}]
	A targeted linear dynamical system $(A,B,T)$ is target controllable if and only if its controllability matrix $[C_TB, C_TAB, C_TA^2B,\ldots,C_TA^{n-1}B]$ is of full rank.
\end{theorem}

The controllability of linear dynamical system has found in the last few years many applications in biology and medicine \cite{art-lsb11}, \cite{Gao:2014aa}, \cite{Guo:2017aa}, \cite{art-kcgp17}, \cite{art-cgkkp18}. In this context the linear dynamical system is an interaction (e.g., signalling) network describing the biological process of interest, and the input is in terms of available drugs or small inhibitors. The difficulty with this application of the concept is that the system is only partially defined, with the majority of the interactions impossible to measure, and thus with the matrices $A$ and $B$ only partially defined. The solution is a structural formulation of the problem, where controllability is defined in terms of the interaction network and not in terms of their precise strength, see, e.g., \cite{art-cgkkp18}. Also, the problem in this context is often given only through the interaction network (the equivalent of matrix $A$ above). The goal is to identify a suitable set of input nodes making the system controllable, i.e., given matrix $A$, the problem is to identify a suitable matrix $B$ such that the linear dynamical system $(A,B)$ is controllable. Furthermore, for applications in medicine, with the input being thought of as drugs delivered to a patient, there are various conditions imposed on matrix $B$, such as having a minimal number of columns (corresponding to minimising the number of drugs), or having the non-zero entries of $B$ only on certain rows (corresponding to selecting the drugs only from a certain set, e.g., FDA-approved drugs, or disease-specific standard drugs). Some examples on applying the controllability problem in medicine are in \cite{art-lsb11}, \cite{art-cgkkp18}, \cite{art-kcgp17}, \cite{MM3networks2020} . Software for solving the controllability problem is available in \cite{netcontrol4biomed}, \cite{NetControlGeneticFrontiers}.

\section{Running example: a reaction system model for breast cancer dynamics}\label{sec:example}
Our running example is a reaction system modelling oncogenic signalling, with a focus on the receptor tyrosine kinase (RTK) signalling network and on the occurrence of uncontrolled proliferation. We follow the Boolean network model proposed in \cite{G.-T.-Zanudo:2018aa} and give it a correspondent in terms of reaction systems.

\begin{figure}
\begin{center}
\begin{tabular}{cc}
	\includegraphics[scale=0.35]{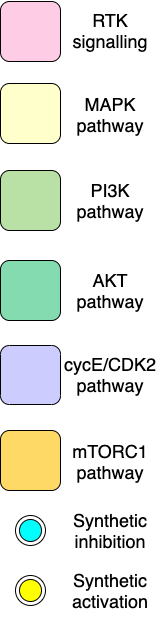}	\ \ \ \ \ \ \
	&
	\includegraphics[scale=0.35]{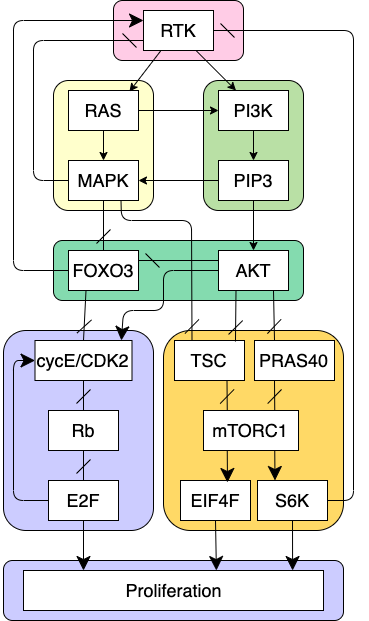}	\\
	 \multicolumn{2}{c}{(a)}
\end{tabular}
\begin{tabular}{ccc}
	\includegraphics[scale=0.27]{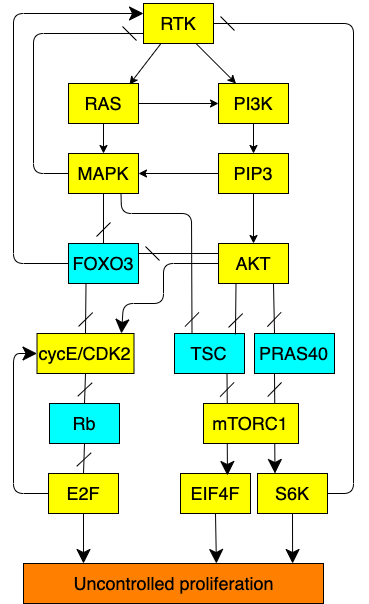}	
	&	
	\includegraphics[scale=0.27]{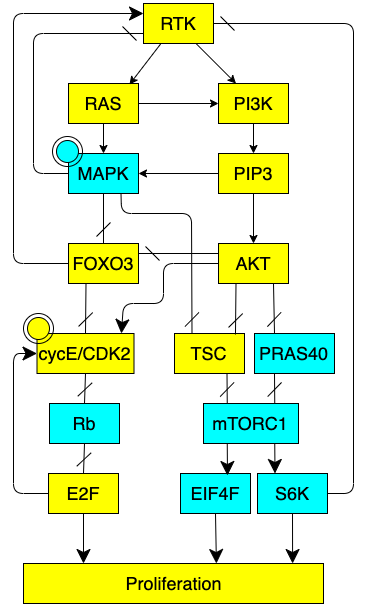}	
	&
	\includegraphics[scale=0.27]{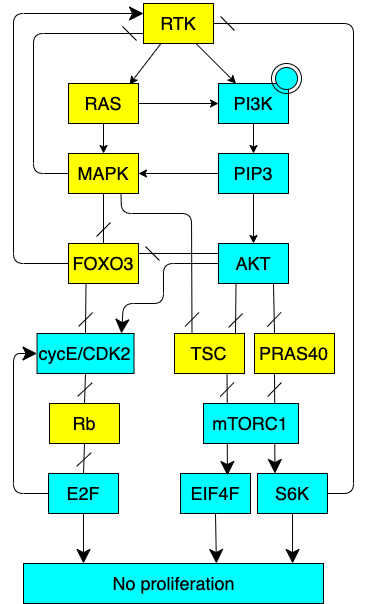}	
	\\
		(b) & (c) & (d) 	
\end{tabular}

\end{center}
	\caption{(a) The signal transduction network of \cite{G.-T.-Zanudo:2018aa} characterising cell proliferation. The $\MAPK$ pathway is in yellow, the $\PIP3$ pathway is in green, the $\AKT$-pathway is in dark green, the $\mTORC1$-pathway is in orange, and the $\cycE$-pathway is in blue. (b) - (d) Network configurations leading to (b) uncontrolled proliferation, (c)  proliferation, and (d) no proliferation. In (b)-(d) a rectangle with blue background denotes an inactive component, and one with yellow denotes an active component. 
	}\label{fig-boolnet}
\end{figure}

The model of \cite{G.-T.-Zanudo:2018aa} is grounded in the context of the breast cancer model of \cite{Gomez-Tejeda-Zanudo:2017aa}. It consists of a simplified version of the RTK signalling network through the MAPK, PI3K, AKT, and mTORC1 pathways, including the cross-talk between them, and several feedback loops. The network includes both oncogene proteins (RAS, PI3K, mTOR), as well as tumour suppressors (Rb, FOXO3). The model is illustrated in Figure \ref{fig-boolnet}(a) (the growth factors $\GF$ are not included in the figure). Each of the key proteins in Figure \ref{fig-boolnet} has a correspondent variable in the Boolean network model, with the update Boolean functions described in Table \ref{table-boolnet}. Variable $X$ gets value $0$ if its correspondent protein is inactive and $1$ if it is active. The outcome of the model is described by variable $\Pro$ which, as an exception, is ternary-valued, with $0$ standing for non-proliferation, $1$ for proliferation, and $2$ for uncontrolled proliferation. Depending on the signals reaching the output node, the model is in one of these three proliferation modes, as described by the update function for variable $\Pro$, $f_{\Pro}=\E2F + (\EIF4F \land \S6K)$. Because of the non-binary nature of this function, we replace the modelling of the proliferation status with two different variables, $\Pro$ that models with true/false the proliferation/non-proliferation status and $\uPro$ that models with true/false the uncontrolled proliferation status.  Their Boolean functions are defined as follows:

\vspace{-4mm}
\begin{align*}
  f_{\Pro}&=(\E2F\land \lnot\EIF4F) \lor (\E2F\land \lnot \S6K) \lor (\lnot \E2F \land \EIF4F \land \S6K), \\
  f_{\uPro}&=\E2F\land \EIF4F \land\S6K.
\end{align*}
\vspace{-4mm}

\begin{longtable}{|p{0.9\textwidth}|}

\endfirsthead

\multicolumn{1}{c}%
{{\bfseries \tablename\ \thetable{} -- continued from previous page}} \\ \hline 
\endhead

\hline \multicolumn{1}{r}{Continued on the next page} \endfoot

\endlastfoot

\caption{The Boolean network model of \cite{G.-T.-Zanudo:2018aa} for oncogenic signalling in disjunctive normal form.}	\label{table-boolnet}\\

\hline
$f_{\RTK}\,=\,(\GF \land \FOXO3) \lor (\GF\land \lnot \S6K \land \lnot \MAPK)$ \\
$f_{\RAS}\,=\,\RTK$\\
$f_{\MAPK}\,=\,\RAS \lor \PIP3$\\
$f_{\PI3K}\,=\,\RTK \lor \RAS$\\
$f_{\PIP3}\,=\,\PI3K$\\
$f_{\FOXO3}\,=\,\lnot\AKT \lor \lnot\MAPK$\\
$f_{\AKT}\,=\,\PIP3$\\
$f_{\cycE}\,=\,(\AKT\land\lnot\FOXO3)\lor\E2F$\\
$f_{\Rb}\,=\,\lnot\cycE$\\
$f_{\E2F}\,=\,\lnot\Rb$\\
$f_{\TSC}\,=\,\lnot\MAPK \lor \lnot\AKT$\\
$f_{\PRAS40}\,=\,\lnot\AKT$\\
$f_{\mTORC1}\,=\,\lnot\TSC \land \lnot\PRAS40$\\
$f_{\EIF4F}\,=\,{\mTORC1}$\\
$f_{\S6K}\,=\,{\mTORC1}$\\
$f_{\Pro}\,=\,(\E2F\land \lnot\EIF4F) \lor (\E2F\land \lnot \S6K) \lor (\lnot \E2F \land \EIF4F \land \S6K)$\\
$f_{\uPro}\,=\,\E2F\land \EIF4F \land\S6K$\\
\hline

\end{longtable}

This model can reach a non-proliferation, a proliferation, and an uncontrolled proliferation status, depending on the different activation statuses of the MAPK-, PI3K-, AKT-, and mTORC1-pathways, see \cite{G.-T.-Zanudo:2018aa} and Figure \ref{fig-boolnet}(b)-(d) on page~\pageref{fig-boolnet}.

The reaction system associated to the Boolean network in Table~\ref{table-boolnet} has as its background set all the variables of the Boolean network and an inhibitor variable associated to each of them: $S=\{\GF$, $\RTK$, $\iota_{\RTK}$, $\RAS$, $\iota_{\RAS}$, $\MAPK$, $\iota_{\MAPK}$, $\PI3K$, $\iota_{\PI3K}$, $\PIP3$, $\iota_{\PIP3}$, $\FOXO3$, $\iota_{\FOXO3}$, $\AKT$, $\iota_{\AKT}$, $\cycE$, $\iota_{\cycE}$, $\Rb$, $\iota_{\Rb}$, $\E2F$, $\iota_{\E2F}$, $\TSC$, $\iota_{\TSC}$, $\PRAS40$, $\iota_{\PRAS40}$, $\mTORC1$, $\iota_{\mTORC1}$, $\EIF4F$, $\iota_{\EIF4F}$, $\S6K$, $\iota_{\S6K}$, $\Pro$, $\iota_{\Pro}$, $\uPro$, $\iota_{\uPro}\}$. The set of reactions is listed in Table~\ref{table-rsmodel}.

In the reaction system model, the equivalent of a component/signal $X$, $X\in\{\GF$, $\RTK$, $\RAS$, $\MAPK$, $\PI3K$, $\PIP3$,  $\FOXO3$, $\AKT$,  $\cycE$, $\Rb$, $\E2F$,  $\TSC$,  $\PRAS40$, $\mTORC1$, $\EIF4F$, $\S6K\}$ being active/inactive is having (not having, resp.) $X$ in the current state. Synthetically activating/inactivating $X$ corresponds to the context adding to the current state $X$ ($\iota_X$, resp.) A non-proliferation configuration corresponds a state that does not include either symbol $\Pro$ or $\uPro$. A proliferation/uncontrolled proliferation configuration corresponds to a state including $\Pro$ ($\uPro$, resp.)

\begin{longtable}{|p{0.9\textwidth}|}

\endfirsthead

\multicolumn{1}{c}%
{{\bfseries \tablename\ \thetable{} -- continued from previous page}} \\ \hline 
\endhead

\hline \multicolumn{1}{r}{Continued on the next page} \endfoot

\endlastfoot

\caption{A reaction systems model for oncogenic signalling.}	\label{table-rsmodel}\\

\hline 
$r_{\RTK}^{(1)}\,=\,(\{\GF, \FOXO3\},\{\iota_{\RTK}\},\{\RTK\})$\\
$r_{\RTK}^{(2)}\,=\,(\{\GF\},\{\S6K, \MAPK, \iota_{\RTK}\},\{\RTK\})$\\
$r_{\RAS}\,=\,(\{\RTK\},\{\iota_{\RAS}\},\{\RAS\})$\\
$r_{\MAPK}^{(1)}\,=\,(\{\RAS\},\{\iota_{\MAPK}\},\{\MAPK\})$\\
$r_{\MAPK}^{(2)}\,=\,(\{\PIP3\},\{\iota_{\MAPK}\},\{\MAPK\})$\\
$r_{\PI3K}^{(1)}\,=\,(\{\RTK \},\{\iota_{\PI3K}\},\{\PI3K\})$\\
$r_{\PI3K}^{(2)}\,=\,(\{\RAS \},\{\iota_{\PI3K}\},\{\PI3K\})$\\
$r_{\PIP3}\,=\,(\{\PI3K \},\{\iota_{\PIP3}\},\{\PIP3\})$\\
$r_{\FOXO3}^{(1)}\,=\,(\emptyset,\{\AKT, \iota_{\FOXO3}\},\{\FOXO3 \})$\\
$r_{\FOXO3}^{(2)}\,=\,(\emptyset,\{\MAPK, \iota_{\FOXO3}\},\{\FOXO3 \})$\\
$r_{\AKT}\,=\,(\{\PIP3 \},\{\iota_{\AKT}\},\{\AKT\})$\\
$r_{\cycE}^{(1)}\,=\,(\{\AKT \},\{\FOXO3, \iota_{\cycE}\},\{\cycE\})$\\
$r_{\cycE}^{(2)}\,=\,(\{\E2F \},\{\ \iota_{\cycE}\},\{\cycE\})$\\
$r_{\Rb}\,=\,(\emptyset,\{\cycE, \iota_{\Rb}\},\{\Rb\})$\\
$r_{\E2F}\,=\,(\emptyset,\{\Rb, \iota_{\E2F}\},\{\E2F\})$\\
$r_{\TSC}^{(1)}\,=\,(\emptyset,\{\MAPK, \iota_{\TSC}\},\{\TSC\})$\\
$r_{\TSC}^{(2)}\,=\,(\emptyset,\{\AKT, \iota_{\TSC}\},\{\TSC\})$\\
$r_{\PRAS40}\,=\,(\emptyset,\{\AKT, \iota_{\PRAS40}\},\{\PRAS40\})$\\
$r_{\mTORC1}\,=\,(\emptyset,\{\TSC, \PRAS40, \iota_{\mTORC1}\},\{\mTORC1\})$\\
$r_{\EIF4F}\,=\,(\{\mTORC1\},\{\iota_{\EIF4F}\},\{\EIF4F\})$\\
$r_{\S6K}\,=\,(\{\mTORC1\},\{\iota_{\S6K}\},\{\S6K\})$\\
$r_{\Pro}^{(1)}\,=\,(\{\E2F \},\{\EIF4F, \iota_{\Pro}\},\{\Pro\})$\\
$r_{\Pro}^{(2)}\,=\,(\{\E2F \},\{\S6K, \iota_{\Pro}\},\{\Pro\})$\\
$r_{\Pro}^{(3)}\,=\,(\{\EIF4F, \S6K \},\{\E2F, \iota_{\Pro}\},\{\Pro\})$\\
$r_{\uPro}\,=\,(\{\E2F,\EIF4F, \S6K \},\{\iota_{\uPro}\},\{\uPro\})$\\
\hline
\end{longtable}

Table~\ref{t-RS-IP} shows an interactive process of our reaction systems model, running with the constant context $\{\GF\}$, and oscillating between proliferation and uncontrolled proliferation.

\newcommand{\iplineskip}{3mm}
\begin{longtable}{p{0.12\textwidth}|p{0.25\textwidth}|p{0.25\textwidth}|p{0.25\textwidth}}

\endfirsthead

\multicolumn{4}{c}%
{{\bfseries \tablename\ \thetable{} -- continued from previous page}} \\ \hline 
\endhead

\hline \multicolumn{4}{r}{Continued on the next page} \endfoot

\hline \hline
\endlastfoot

\caption{An interactive process of the reaction systems model for oncogenic signalling.  
}	
\label{t-RS-IP}\\

\hline \hline
\textbf{Context}	&	$\GF$&	$\GF$&	$\GF$\\
\hline 
\textbf{State}		&	$S_1=\{\RTK$, $\FOXO3$, $\Rb$, $\E2F$, $\TSC$, $\PRAS40$, $\mTORC1\}$ & $S_2=\{\RTK$, $\RAS$, $\PI3K$, $\FOXO3$, $\cycE$, $\Rb$, $\TSC$, $\PRAS40$, $\EIF4F$, $\S6K$, $\Pro\}$ & $S_3=\{\RTK$, $\RAS$, $\MAPK$, $\PI3K$, $\PIP3$, $\FOXO3$, $\TSC$, $\PRAS40$, $\Pro\}$\\
\hline 
\textbf{Status}		&	\cellcolor{cyan}\emph{No proliferation}		&	\cellcolor{yellow}\emph{Proliferation}		&	\cellcolor{yellow}\emph{Proliferation}\\

\multicolumn{4}{c}{}\\

\textbf{Context}	&	$\GF$&	$\GF$&	$\GF$	\\
\hline 
\textbf{State}		&	$S_4=\{\RTK$, $\RAS$, $\MAPK$, $\PI3K$, $\PIP3$, $\FOXO3$, $\AKT$, $\Rb$, $\E2F$, $\TSC$, $\PRAS40\}$		&	$S_5=\{\RTK$, $\RAS$, $\MAPK$, $\PI3K$, $\PIP3$, $\AKT$, $\cycE$, $\Rb$, $\Pro\}$	 &	$S_6=\{\RAS$, $\MAPK$, $\PI3K$, $\PIP3$, $\AKT$, $\cycE$, $\mTORC1\}$\\
\hline 
\textbf{Status}		&	\cellcolor{cyan}\emph{No proliferation}		&	\cellcolor{yellow}\emph{Proliferation}		&	\cellcolor{cyan}\emph{No proliferation}	\\


\textbf{Context}	&	$\GF$&	$\GF$&	$\GF$	\\
\hline 
\textbf{State}		&	$S_7=\{\MAPK$, $\PI3K$, $\PIP3$, $\AKT$, $\cycE$, $\E2F$, $\mTORC1$, $\EIF4F$, $\S6K\}$	&	$S_8=\{\MAPK$, $\PIP3$, $\AKT$, $\cycE$, $\E2F$, $\mTORC1$, $\EIF4F$, $\S6K$, $\uPro\}$	&	$S_9=\{\MAPK$, $\AKT$, $\cycE$, $\E2F$, $\mTORC1$, $\EIF4F$, $\S6K$, $\uPro\}$\\
\hline 
\textbf{Status}		&	\cellcolor{cyan}\emph{No proliferation}		&	\cellcolor{orange}\emph{Uncontr. prolif.}		&	\cellcolor{orange}\emph{Uncontr. prolif.}		\\

\multicolumn{4}{c}{}\\[\iplineskip]

\textbf{Context}	&	$\GF$&	$\GF$&	$\GF$	\\
\hline 
\textbf{State}		&	$S_{10}=\{\cycE$, $\E2F$, $\mTORC1$, $\EIF4F$, $\S6K$, $\uPro\}$		&	$S_{11}=\{\FOXO3$, $\cycE$, $\E2F$, $\TSC$, $\PRAS40$, $\mTORC1$, $\EIF4F$, $\S6K$, $\uPro\}$ 	&	$S_{12}=\{\RTK$, $\FOXO3$, $\cycE$, $\E2F$, $\TSC$, $\PRAS40$, $\EIF4F$, $\S6K$, $\uPro\}$\\
\hline 
\textbf{Status}		&	\cellcolor{orange}\emph{Uncontr. prolif.}		&	\cellcolor{orange}\emph{Uncontr. prolif.}		&	\cellcolor{orange}\emph{Uncontr. prolif.}		\\
 
\multicolumn{4}{c}{}\\[\iplineskip]

\textbf{Context}	&	$\GF$&	$\GF$&	$\GF$	\\
\hline 
\textbf{State}		&	$S_{13}=\{\RTK$, $\RAS$, $\PI3K$, $\FOXO3$, $\cycE$, $\E2F$, $\TSC$, $\PRAS40$, $\uPro\}$	&	$S_{14}=\{\RTK$, $\RAS$, $\MAPK$, $\PI3K$, $\PIP3$, $\FOXO3$, $\cycE$, $\E2F$, $\TSC$, $\PRAS40$, $\Pro\}$	&	$S_{15}=\{\RTK$, $\RAS$, $\MAPK$, $\PI3K$, $\PIP3$, $\FOXO3$, $\AKT$, $\cycE$, $\E2F$, $\TSC$, $\PRAS40$, $\Pro\}$\\
\hline 
\textbf{Status}		&	\cellcolor{orange}\emph{Uncontr. prolif.}		&	\cellcolor{yellow}\emph{Proliferation}		&	\cellcolor{yellow}\emph{Proliferation}		\\

\multicolumn{4}{c}{}\\[\iplineskip]

\textbf{Context}	&	$\GF$&	$\GF$&	$\GF$	\\
\hline 
\textbf{State}		&	$S_{16}=\{\RTK$, $\RAS$, $\MAPK$, $\PI3K$, $\PIP3$, $\AKT$, $\cycE$, $\E2F$, $\Pro\}$		&	$S_{17}=\{\RAS$, $\MAPK$, $\PI3K$, $\PIP3$, $\AKT$, $\cycE$, $\E2F$, $\mTORC1$, $\Pro\}$		&	$S_{18}=\{\MAPK$, $\PI3K$, $\PIP3$, $\AKT$, $\cycE$, $\E2F$, $\mTORC1$, $\EIF4F$, $\S6K$, $\Pro\}$\\
\hline 
\textbf{Status}		&	\cellcolor{yellow}\emph{Proliferation}		&	\cellcolor{yellow}\emph{Proliferation}		&	\cellcolor{yellow}\emph{Proliferation}		\\

\multicolumn{4}{c}{}\\[\iplineskip]

\textbf{Context}	&	$\GF$ & \hspace{11mm} {\Large $\cdots$} & \hspace{11mm} {\Large $\cdots$}\\
\hline 
\textbf{State}		&	$S_{19}=\{\MAPK$, $\PIP3$, $\AKT$, $\cycE$, $\E2F$, $\mTORC1$, $\EIF4F$, $\S6K$, $\uPro\} = S_8$  &  \phantom{aaa aaa aaa aaa aaa}  \phantom{aaa aaa aaa aaa aaa}  \phantom{aaa aa} {\huge $\dots$}   &  \phantom{aaa aaa aaa aaa aaa}  \phantom{aaa aaa aaa aaa aaa}  \phantom{aaa aa} {\huge $\dots$} \\
\hline 
\textbf{Status}		&	\cellcolor{orange}\emph{Uncontr. prolif.}   &  \hspace{11mm} {\Large $\cdots$}  & \hspace{11mm} {\Large $\cdots$} \\
\end{longtable}

\section{Controllability of reaction systems}
\label{sec:control}

The basic concept of controllability of a reaction system \(\mathcal{A} = (S, A)\) is that for any \(X, Y \subseteq S\), \(X \neq Y\), there is an interactive process of \(\mathcal{A}\) starting in $X$ and ultimately reaching $Y$. More exactly, the controllability problem \((\mathcal{A}, X, Y)\) is finding a context sequence \(\mathcal{C} = (C_i)_{0 \leq i \leq n}\) for \(\mathcal{A}\) such that the interactive process generated by \(\mathcal{C}\) starts in $X$ and ends in \(Y\).  $\mathcal{A}$ is said to be controllable if the controllability problem $(\mathcal A, X, Y)$ has a solution for any pair $X, Y \subseteq S$, $X \neq Y$.

Similarly as in the case of linear dynamical systems, in the absence of constraints on the context sequences, the controllability problem has a trivial solution: for any $X,Y\subseteq S$, the control sequence \(\mathcal{\gamma}_{X \to Y} = X, S, Y\) leads to an interactive process starting in $X$, going to the empty state (since all reactions are disabled through getting the full set $S$ as a context), and then going to $Y$:
\begin{center}
\begin{tabular}{|p{0.15\textwidth}||p{0.2\textwidth}|p{0.2\textwidth}|p{0.2\textwidth}|} 
\hline \textbf{Context}  & $X$ & $S$ & $Y$  \\
\textbf{Result}  & $\emptyset$ & $\res_A(X)$ & $\emptyset=\res_A(S)$ \\
\hline \textbf{State} & $X$ & $S$ & $Y$ \\
\hline 
\end{tabular}
\end{center}

Instead, we define controllability of reaction systems as follows.

\begin{definition}
	Let $\mathcal{A}=(S,A)$ be a reaction system. 
	\begin{itemize}
	\item 	For some nonnegative integer $n$ with $0\leq n<|S|$, we say that $\mathcal{A}$ is $n$-controllable if for any $X,Y\subseteq S$, with $Y$ reachable in $\mathcal{A}$, there is a context sequence consisting of sets of cardinality at most $n$ generating an interactive process starting in $X$ and ending in $Y$.
	
	\item For some $I\subseteq S$, we say that $\mathcal{A}$ is $I$-controllable if for any $X,Y\subseteq S$, with $Y$ reachable in $\mathcal{A}$, there is a context sequence consisting of subsets of $I$ generating an interactive process starting in $X$ and ending in $Y$.
	\end{itemize}
\end{definition}


We define two versions of the controllability problem for reaction systems as follows. Let $\mathcal{A}=(S,A)$ be a reaction system.

\begin{description}
	\item [Problem $\mathcal{C}(\mathcal{A},n)$] For some nonnegative integer $n$ with $0\leq n<|S|$, decide if $\mathcal{A}$ is $n$-controllable. Also, find the smallest such $n$.
	
	\item [Problem $\mathcal{C}(\mathcal{A},I)$] For some $I\subseteq S$, decide if $\mathcal{A}$ is $I$-controllable. Also, find a minimal (with respect to inclusion) such set $I$.
\end{description}


\begin{example}
	The interactive process in Table \ref{t-control-ex} is an example on how to induce a change of state from the attractor state $S_{18}$ (with uncontrolled proliferation) in Table \ref{t-RS-IP} to state $\{\RTK$, $\RAS$, $\MAPK$, $\FOXO3$, $\cycE$, $\E2F$, $\TSC$, $\PRAS40$, $\Pro\}$ (with Proliferation). We use just one additional symbol in the context sequence, $\iota_{\PI3K}$, inhibiting all reactions producing $\PI3K$.

\vspace{-1mm}

\begin{longtable}{p{0.12\textwidth}|p{0.25\textwidth}|p{0.25\textwidth}|p{0.25\textwidth}}

\endfirsthead

\multicolumn{4}{c}%
{{\bfseries \tablename\ \thetable{} -- continued from previous page}} \\ \hline 
\endhead

\hline \multicolumn{4}{r}{Continued on the next page} \endfoot

\hline \hline
\endlastfoot

\caption{An interactive process of the reaction systems model for oncogenic signalling switching from an uncontrolled proliferation status to a proliferation status.}	\label{t-control-ex}\\

\hline \hline
\textbf{Context}	&	\cellcolor{lightgray}$\{\GF\}$&	$\{\GF, \iota_{\PI3K}\}$&	$\{\GF, \iota_{\PI3K}\}$\\
\hline 
\textbf{State}		&	\cellcolor{lightgray}$X_0 = S_{19}=\{\MAPK$, $\PIP3$, $\AKT$, $\cycE$, $\E2F$, $\mTORC1$, $\EIF4F$, $\S6K$, $\uPro\}$ & $X_{1}=\{\MAPK$, $\AKT$, $\cycE$, $\E2F$, $\mTORC1$, $\EIF4F$, $\S6K$, $\uPro\}$ & $X_{2}=\{\cycE$, $\E2F$, $\mTORC1$, $\EIF4F$, $\S6K$, $\uPro\}$\\
\hline 
\textbf{Status}		&	\cellcolor{lightgray!60!orange!50}\emph{Uncontr. prolif.}		&	\cellcolor{orange}\emph{Uncontr. prolif.}		&	\cellcolor{orange}\emph{Uncontr. prolif.}\\

\multicolumn{4}{c}{}\\

\textbf{Context}	&	$\{\GF, \iota_{\PI3K}\}$&	$\{\GF, \iota_{\PI3K}\}$&	$\{\GF, \iota_{\PI3K}\}$	\\
\hline 
\textbf{State}		&	$X_{3}=\{\FOXO3$, $\cycE$, $\E2F$, $\TSC$, $\PRAS40$, $\mTORC1$, $\EIF4F$, $\S6K$, $\uPro\}$		&	$X_{4}=\{\RTK$, $\FOXO3$, $\cycE$, $\E2F$, $\TSC$, $\PRAS40$, $\EIF4F$, $\S6K$, $\uPro\}$	 &	$X_{5}=\{\RTK$, $\RAS$, $\FOXO3$, $\cycE$, $\E2F$, $\TSC$, $\PRAS40$, $\uPro\}$\\
\hline 
\textbf{Status}		&	\cellcolor{orange}\emph{Uncontr. prolif.}		&	\cellcolor{orange}\emph{Uncontr. prolif.}			&	\cellcolor{orange}\emph{Uncontr. prolif.} \\

\multicolumn{4}{c}{}\\

\textbf{Context}	&	$\{\GF, \iota_{\PI3K}\}$&	$\{\GF, \iota_{\PI3K}\}$	\\
\hline 
\textbf{State}		&	$X_{6}=\{\RTK$, $\RAS$, $\MAPK$, $\FOXO3$, $\cycE$, $\E2F$, $\TSC$, $\PRAS40$, $\Pro\}$		&	$X_{7} = X_6$ 	\\
\hline 
\textbf{Status}		&	\cellcolor{yellow}\emph{Proliferation}		&	\cellcolor{yellow}\emph{Proliferation}				
\end{longtable}
\end{example}

\vspace{-7mm}

\section{Target controllability of reaction systems}
\label{sec:target_control}

The concept of target controllability of reaction systems is focused on a given set of target nodes $T\subseteq S$ of a reaction system $\mathcal{A}=(S,A)$.  The objective in this problem is to be able, through a suitable context sequence, to transition between any two subsets of $T$. The caveat here is that, with the focus set on $T$ only, elements from $S\setminus T$ may be present arbitrarily in the interactive process, including in the initial and final states. We formalise this concept as follows.

\begin{definition}
	Let $\mathcal{A}=(S,A)$ be a reaction system and $T\subseteq S$ the set of targets. 
	\begin{itemize}
	\item 	For some nonnegative integer $n$ with $0\leq n<|S|$, we say that $(\mathcal{A},T)$ is $n$-target controllable if for any $X,Y\subseteq T$, with $Y$ or a superset of it reachable in $\mathcal{A}$, there is a context sequence consisting of sets of cardinality at most $n$ generating an interactive process starting in a state $W_0$ with $W_0\cap T=X$ and ending in a state $W_r$ with $W_r\cap T=Y$, for some $r\geq 0$.
	
	\item For some $I\subseteq S$, we say that $(\mathcal{A},T)$ is $I$-target controllable if for any $X,Y\subseteq T$, with $Y$ or a superset of it reachable in $\mathcal{A}$, there is a context sequence consisting of subsets of $I$ generating an interactive process starting in a state $W_0$ with $W_0\cap T=X$ and ending in a state $W_r$ with $W_r\cap T=Y$, for some $r\geq 0$.
	\end{itemize}
\end{definition}

We define two versions of the target controllability problem for reaction systems as follows. Let $\mathcal{A}=(S,A)$ be a reaction system and $T\subseteq S$ the set of targets.

\begin{description}
	\item [Problem $\mathcal{TC}(\mathcal{A},T,n)$] For some nonnegative integer $n$ with $0\leq n<|S|$, decide if $(\mathcal{A},T)$ is $n$-target controllable. Also, find the smallest such $n$.
	
	\item [Problem $\mathcal{TC}(\mathcal{A},T,I)$] For some $I\subseteq S$, decide if $(\mathcal{A},T)$ is $I$-target controllable. Also, find a minimal (with respect to inclusion) such set $I$.
\end{description}


\begin{example}
	A natural target of our running example is the proliferation symbol. Consider for example driving the model away from the state $S_{18}$ in Table~\ref{t-RS-IP}, characterised by uncontrolled proliferation, and into a steady state or an attractor consisting of states characterised by proliferation or no proliferation. There are several ways of doing that with constant context sequences that still include the $\GF$ symbol, including the following options:
	\begin{itemize}
		\item the context set $\{\GF, \PRAS40\}$ drives the model into an attractor consisting of 10 states, all of them with $\Pro$;
		\item the context set $\{\GF, \iota_{\cycE}\}$ drives the model into an attractor consisting of 11 states, 6 with $\Pro$, and 5 without;
		\item the context set $\{\GF, \iota_{\cycE}, \PRAS40\}$ drives the model into an attractor consisting of 10 states, all of them with no proliferation;
		\item the context set $\{\GF, \iota_{\PI3K}, \iota_{\cycE}\}$ drives the model into a steady state with no proliferation, with the interactive process shown in Table \ref{t-target-ex}.
	\end{itemize}

\vspace{-2mm}
	
\begin{longtable}{p{0.12\textwidth}|p{0.25\textwidth}|p{0.25\textwidth}|p{0.25\textwidth}}

\endfirsthead

\multicolumn{4}{c}%
{{\bfseries \tablename\ \thetable{} -- continued from previous page}} \\ \hline 
\endhead

\hline \multicolumn{4}{r}{Continued on the next page} \endfoot

\hline \hline
\endlastfoot

\caption{An interactive process of the reaction systems model for oncogenic signalling switching from uncontrolled proliferation to a steady state with no proliferation.}	\label{t-target-ex}\\

\hline \hline
\textbf{Context}	&	\cellcolor{lightgray}$\{\GF\}$ &	$\{\GF, \iota_{\PI3K}$, $\iota_{\cycE}\}$&	$\{\GF, \iota_{\PI3K}$, $\iota_{\cycE}\}$\\
\hline 
\textbf{State}		&	\cellcolor{lightgray}$Y_0 = S_{19}=\{\MAPK$, $\PIP3$, $\AKT$, $\cycE$, $\E2F$, $\mTORC1$, $\EIF4F$, $\S6K$, $\uPro\}$ &  $Y_{1}=\{\MAPK$, $\AKT$, $\E2F$, $\mTORC1$, $\EIF4F$, $\S6K$, $\uPro\}$ & $Y_{2}=\{\Rb$, $\E2F$, $\mTORC1$, $\EIF4F$, $\S6K$, $\uPro\}$\\
\hline 
\textbf{Status}		&	\cellcolor{lightgray!60!orange!50}\emph{Uncontr. prolif.}		&	\cellcolor{orange}\emph{Uncontr. prolif.}		&	\cellcolor{orange}\emph{Uncontr. prolif.}\\

\multicolumn{4}{c}{}\\

\textbf{Context}	&	$\{\GF, \iota_{\PI3K}$, $\iota_{\cycE}\}$&	$\{\GF, \iota_{\PI3K}$, $\iota_{\cycE}\}$	& $\{\GF, \iota_{\PI3K}$, $\iota_{\cycE}\}$\\
\hline 
\textbf{State}		&	$Y_{3}=\{\FOXO3$, $\Rb$, $\TSC$, $\PRAS40$, $\mTORC1$, $\EIF4F$, $\S6K$, $\uPro\}$		&	$Y_{4}=\{\RTK$, $\FOXO3$, $\Rb$, $\TSC$, $\PRAS40$, $\EIF4F$, $\S6K$, $\Pro\}$		&	$Y_{5}=\{\RTK$, $\RAS$, $\FOXO3$, $\Rb$, $\TSC$, $\PRAS40$, $\Pro\}$ \\
\hline 
\textbf{Status}		&	\cellcolor{orange}\emph{Uncontr. prolif.}		&	\cellcolor{yellow}\emph{Proliferation}		&	\cellcolor{yellow}\emph{Proliferation}  \\

\multicolumn{4}{c}{}\\[1mm]

\textbf{Context}	&	$\{\GF, \iota_{\PI3K}$, $\iota_{\cycE}\}$&	$\{\GF, \iota_{\PI3K}$, $\iota_{\cycE}\}$	\\
\hline 
\textbf{State}		&	$Y_{6}=\{\RTK$, $\RAS$, $\MAPK$, $\FOXO3$, $\Rb$, $\TSC$, $\PRAS40\}$		&	$Y_7 = Y_6$		\\
\hline 
\textbf{Status}		&	\cellcolor{cyan}\emph{No proliferation}		&	\cellcolor{cyan}\emph{No proliferation}		\end{longtable}
\end{example}

\section{Complexity results for the controllability of reaction systems}
\label{sec:complexity}

In this section we show that while unrestricted controllability is
trivial, it may become very complex if additional requirements are
imposed on context sequences.  This observation is also valid for
target controllability.

\subsection{Controllability is \PSPACE-hard}
\label{sec:controllability-pspace-hard}
We will show that $I$-controllability is at least as hard as
reachability, which is \PSPACE-complete~\cite{ReachabilityRS}.

\begin{theorem}\label{prop:i-pspace-hard}
  Let $\mathcal A = (S,A)$ be a reaction system and $I \subseteq S$.
  The problem $\mathcal C(\mathcal A, I)$ of deciding whether
  $\mathcal A$ is $I$-controllable is \PSPACE-hard.
\end{theorem}
\begin{proof}
  We will consider the particular case of $I$-controllability in which
  $I = \emptyset \subseteq S$.  Deciding that $\mathcal A$ is
  $\emptyset$-controllable is equivalent to deciding whether, for any
  pair of subsets $X, Y \subseteq S$, $Y$ is reachable from $X$.
  Since reachability for reaction systems is \PSPACE-complete, we
  conclude that $\emptyset$-controllability and therefore
  $I$-controllability of reaction systems is \PSPACE-hard.
\end{proof}

Similarly, reachability can be reduced to $n$-controllability.

\begin{theorem}\label{prop:n-pspace-hard}
  Let $\mathcal A = (S,A)$ be a reaction system and $0 \leq n < S$.
  The problem $\mathcal C(A, n)$ of deciding whether $\mathcal A$ is
  $n$-controllable is \PSPACE-hard.
\end{theorem}
\begin{proof}
  As in the proof of Theorem~\ref{prop:i-pspace-hard}, we consider
  the particular case $n = 0$, which only allows empty sets as
  contexts.  In this case, $n$-controllability of $\mathcal{A}$ is
  equivalent to deciding whether, for any pair of subsets
  $X, Y \subseteq S$, $Y$ is reachable from $S$.  This implies that
  $n$-controllability is \PSPACE-hard.
\end{proof}

\subsection{Target controllability is \PSPACE-hard}
\label{sec:target-pspace-hard}
Target controllability can be reduced to ``full'' controllability by
allowing any species to be a control target: $T = S$.  This directly
implies that $I$-target controllability and $n$-target controllability
are both \PSPACE-hard.

\begin{corollary}
  Let $\mathcal A = (S,A)$ be a reaction system and
  $I, T \subseteq S$.  The problem $\mathcal{TC}(\mathcal A, T, I)$ of
  deciding whether $\mathcal A$ is $I$-controllable is \PSPACE-hard.
\end{corollary}

\begin{corollary}
  Let $\mathcal A = (S,A)$ be a reaction system, $I \subseteq S$, and
  $0 \leq n \leq |S|$.  The problem $\mathcal{TC}(\mathcal A, T, n)$
  of deciding whether $\mathcal A$ is $n$-controllable is \PSPACE-hard.
\end{corollary}

These results are degenerate in the sense that they focus on the
situations in which control inputs are disallowed: $I = \emptyset$ or
$n = 0$.  We will now show that $I$-controllability for reaction
systems is \PSPACE-hard even when some control inputs \emph{must} be
provided.

\begin{theorem}
  Let $\mathcal A = (S,A)$ be a reaction system and
  $I, T \subseteq S$, such that $I \neq \emptyset$.  The problem
  $\mathcal{TC}(\mathcal A, T, I)$ of deciding whether $\mathcal A$ is
  $I$-controllable is \PSPACE-hard.
\end{theorem}
\begin{proof}
  Consider an arbitrary reaction system $\mathcal{A} = (S, A)$, an
  extension of the background set $S' \supsetneq S$, and the reaction
  system $\mathcal{A}' = (S', A)$ over the extended alphabet, but with
  the same reactions as $\mathcal{A}$.  Let $I = S' \setminus S$ and
  $T = S$.  Then $(\mathcal{A}', T)$ is $I$-target controllable if and
  only if $\mathcal{A}$ is $\emptyset$-controllable.

  Indeed, by construction the species from $I = S' \setminus S$ have
  no influence on the reactions in $A$.  Therefore, given any pair of
  sets $X, Y \subseteq S = T$, if $\mathcal{A}'$ can reach $Y$ from
  $X$ with contexts from $I$, it can reach $Y$ from $X$ with an empty
  context sequence.  Since $\mathcal{A}'$ and $\mathcal{A}$ have
  exactly the same reactions, this also means that if $\mathcal{A}'$
  can reach $Y$ from $X$, $\mathcal{A}$ can reach $Y$ from $X$ as
  well.  Conversely, if $\mathcal{A}$ can reach $Y$ from $X$ with a
  sequence of empty contexts, $\mathcal{A}'$ can reach $Y$ from $X$
  with a sequence of empty contexts, or indeed with any sequence of
  contexts from $I$ of the same length.

  We conclude that, for non-empty $I$, $I$-target controllability is
  at least as hard as $I$-controllability, which proves the statement
  of the theorem.
\end{proof}

The previous proof relies on restricting the set $I$ to symbols not
actually appearing in the reactions.  We will now show that $I$-target
controllability is \PSPACE-hard even when contexts are allowed to
inject symbols appearing in the reactants and the inhibitors of
reactions.

Before stating and proving this result, we formulate the following
helper notion.  Consider the reaction system $\mathcal{A} = (S, A)$
and take an extension of the background set $S' \supseteq S$.  We
define the \emph{nonce-extension} $\nonce_{S \to S'}(A)$ of $A$ from
$S$ to $S'$ in the following way:
\[
  \nonce_{S \to S'}(A) = \{(R \cup R', I \cup I', P)
  \mid a \in A, R' \subseteq S' \setminus S,
  I' \subseteq S' \setminus S\}.
\]
Note that $A \subseteq \nonce_{S \to S'}(A)$, because
$\emptyset \subseteq S' \setminus S$.

\medskip

\begin{theorem}
  Let $\mathcal A = (S,A)$ be a reaction system and
  $E, T \subseteq S$ such that $E \cap \rsc(A) \neq \emptyset$.  The problem $\mathcal{TC}(\mathcal A, T, E)$ of deciding
  whether $\mathcal A$ is $E$-controllable is \PSPACE-hard.
\end{theorem}
\begin{proof}
  Consider a reaction system $\mathcal{A} = (S, A)$ and an extension
  of the background set $S' \supsetneq S$.  Construct now a new
  reaction system $\mathcal{A}' = (S', A')$ with the properties
  $A \subseteq A' \subseteq \nonce_{S \to S'}(A)$.  Let
  $E = S' \setminus S$ and $T = S$.  Then $(\mathcal{A}', T)$ is
  $E$-target controllable if and only if $\mathcal{A}$ is
  $\emptyset$-controllable.  Stronger yet, for any $Z \subseteq S'$,
  $\res_{\mathcal A'}(Z) = \res_{\mathcal A}(Z \cap S)$.  Indeed,
  suppose reaction $a' = (R, I, P) \in A'$ is enabled by $Z$,
  i.e.~$R \subseteq Z$ and $I \cap Z = \emptyset$.  Then, trivially,
  the reaction $a = (R \cap S, I \cap S, P)$ is enabled by $Z \cap S$,
  and $a \in A$ by construction of $A'$.  On the other hand, if
  reaction $a \in A$ is enabled by $Z \cap S$, then it is also enabled
  by $Z$, because $(Z \setminus S) \cap (R \cup I) = \emptyset$.
\end{proof}

\section{Conclusions}

The controllability problem is of high interest in dynamical systems, having as its aim the ability to change the system's configuration through well chosen sequences of external interventions. We initiated in this paper the study of controllability for reaction systems. The reaction systems framework has all the key ingredients necessary for a natural definition of the controllability problem: system dynamics through interactive processes, state transitions, external interventions through context sequences. We defined several natural variants of the controllability problem for reaction systems.

We introduced the first reaction system-based oncogenic signalling model in the literature, a model that we believe will be of independent interest to the reaction systems community. The model includes several of the best studied cancer signalling pathways and follows their interplay leading to tumour proliferation, both in the case of external growth factor signals, as well as in their absence. We used this example to show how much diversity of options there is in the concept of controllability for reaction systems. The complexity of dynamics shown through this example anticipated the computational complexity results we proved in this article, showing that the controllability problem is \PSPACE-hard. 

Several topics of interest remain to be explored around controllability of reaction systems, for example, and in no particular order:
\begin{enumerate}
\item 	Are there formulations of the problem that are computationally easy (in the sense of computational complexity theory), perhaps based on minimal reaction systems? 
\item Define and study the concept of stable controllability, where a constant (or an ultimately constant) context sequence leads to the desired state, that moreover is a steady state of the reaction system with the given constant context sequence. 
\item Find efficient heuristics for the controllability of reaction systems, identifying (not necessarily optimal) context sequences solving a given controllability problem. 
\end{enumerate}
We believe that these topics should give further insight into the
potential of reaction systems as a qualitative framework for
biomodelling.

\section*{Funding}
Sergiu Ivanov was partially supported by Computer Science Network of Paris \^{I}le-de-France Region, project AAP DIM RFSI 2018-03.
Ion Petre was partially supported by the Romanian National Authority for Scientific Research and Innovation (POC grant P\_37\_257 and PED grant 2391).

\section*{Data availability}
All data analysed in this study (the reaction systems model and the interactive processes) is described in full in the article.

\section*{Conflict of interest}

The authors declare that they have no conflict of interest.

\bibliographystyle{plain}
\bibliography{rs-control}

\end{document}